\documentclass[12pt]{extarticle}
\usepackage{graphicx}
\usepackage{grffile}
\usepackage{framed}
\usepackage[normalem]{ulem}
\usepackage{amsmath}
\usepackage{amsthm}
\usepackage{amssymb}
\usepackage{amsfonts}
\usepackage{setspace}
\usepackage{enumerate}
\usepackage[utf8]{inputenc}
\usepackage[top=1 in,bottom=1in, left=1 in, right=1 in]{geometry}
\usepackage{subcaption}
\usepackage[ruled,vlined]{algorithm2e}
\usepackage{authblk}

\usepackage{multirow}
\usepackage{mathrsfs} 

\theoremstyle{definition}
\newtheorem{theorem}{Theorem}
\newtheorem{lemma}[theorem]{Lemma}

\DeclareMathOperator*{\argmin}{arg\,min}
\usepackage{color}
\newcommand\revision[1]{{#1}}
\providecommand{\keywords}[1]{\textbf{\small{Keywords:}} #1}

\title{Machine Learning and Polymer Self-Consistent Field Theory in Two Spatial Dimensions}

\author[1\thanks{corresponding author, yxscience@gmail.com}\thanks{work is done while at UC Santa Barbara}]{Yao Xuan}
\author[2]{Kris T. Delaney}
\author[1]{Hector D. Ceniceros}
\author[2,3]{Glenn H. Fredrickson}
\affil[1]{Department of Mathematics, University of California, Santa Barbara, California 93106,
USA}
\affil[2]{Materials Research Laboratory, University of California, Santa Barbara, California 93106, USA}
\affil[3]{Departments of Materials and Chemical Engineering, University of California, Santa Barbara, California 93106, USA}
\begin{document}

\maketitle
\begin{abstract}
A computational framework that leverages data from self-consistent field theory simulations with deep learning to accelerate the exploration of parameter space for block copolymers is presented. This is a substantial two-dimensional extension of the framework introduced in~\cite{xuan2021deep}. Several innovations and improvements are proposed. (1) A Sobolev space-trained, convolutional neural network (CNN) is employed  to handle the exponential dimension increase of the discretized, local average monomer density fields and to strongly enforce both spatial translation and rotation invariance of the predicted, field-theoretic intensive Hamiltonian. (2) A generative adversarial network (GAN) is introduced to efficiently and accurately predict saddle point, local average monomer density fields without resorting to gradient descent methods that employ the training set. This GAN approach yields important savings of both memory and computational cost.  (3) The proposed machine learning framework is successfully applied to 2D cell size optimization as a clear illustration of its broad potential to accelerate the exploration of parameter space for discovering polymer nanostructures.  Extensions to three-dimensional phase discovery appear to be feasible.

\end{abstract}

\begin{section}{Introduction}\label{sec:2d_introduction}
Numerical simulations based on self-consistent field theory (SCFT) are a powerful tool to study the energetics and structures of polymer phases~\cite{fredrickson2006equilibrium,Matsen2007,Schmid_1998}. However, the high cost of these direct computations, involving the repeated solution of multiple modified diffusion (Fokker-Planck) equations~\cite{fredrickson2006equilibrium,CF2004,stasiak2011efficiency}, hinders the scalability of SCFT for its application in polymer phase discovery. 

Machine learning (ML), which obtains an approximate input-to-output map from data,  can substantially reduce (after training) the computational cost of evaluating quantities of interest. Consequently, there has been increasing interest to combine ML with traditional polymer SCFT simulations to speed up the exploration of parameter space. 
Most of the work to date has been focused on separate stages of SCFT computation and/or downstream tasks. For example, Wei, Jiang, and Shi~\cite{wei2018machine} trained a neural network to the solution of SCFT's modified diffusion equations to accelerate the computations of the mean fields.
Nakamura~\cite{nakamura2020phase} proposed to predict the type of polymer phase by using a neural network with a theory-embedded layer that captures the characteristic phase features via coarse-grained mean-field theory. Arora, Lin, Rebello, Av-Ron, Mochigase and Olsen~\cite{arora2021random} employed random forests to predict diblock copolymer phase behavior.  On the other hand, our approach in~\cite{xuan2021deep} is a complete ML framework for full SCFT simulation via deep learning to obtain both accurate approximations of the 
 field-theoretic intensive Hamiltonian
(also called effective Hamiltonian and representing the intensive free energy in the mean-field approximation) and predictions of the (local average) monomer density field at SCFT saddle points. 
This initial ML framework was developed in the one-dimensional spatial setting to test this  combination of ML and SCFT. The framework 
in~\cite{xuan2021deep} is a two-step computational strategy: (1) train a (fully connected) deep neural network in Sobolev space, which predicts the effective Hamiltonian and its gradient simultaneously, and enforces translation invariance  and (2) leverage the gradient descent method on the Sobolev space-trained deep network (the surrogate effective Hamiltonian) to find an approximation of the monomer density field at a saddle point with initial density selected from the training set. 

Here, we propose an efficient and superior end-to-end ML framework, both in memory and computational cost, than that developed in the one-dimensional spatial setting. Thus, this is not just a two-dimensional extension of~\cite{xuan2021deep}.

 The increase of the spatial dimension introduces two significant challenges.  
 First, 
 \revision{the size of the input discrete monomer density field increases like $n^d$ where $n$ is the number of field values (values at grid points) per dimension and $d$ is the spatial dimension}.  Second, the effective Hamiltonian must be invariant under both 
 translation {\em and rotation} transformations of the monomer density field \revision{(and the polymer cell)}. 
 To overcome these salient difficulties,  we introduce a new ML-approach for the complete prediction of the effective Hamiltonian and the average monomer density. This new framework has the following main 
 innovations:
\begin{enumerate}
    \item We view the monomer density fields as 2D images and replace the fully connected neural network with a more efficient
convolutional neural network (CNN)~\cite{cnnorigin} trained in Sobolev space to guarantee accurate approximations of both the effective Hamiltonian and its gradient in a vicinity of saddle points. 
 CNN's have proved to be a powerful tool for many image-related problems~\cite{VALUEVA2020232}.
 \item We use a rotation-invariant representation of the polymer cell parameters and design a CNN architecture based on circular padding to achieve strong rotation and translation invariance. This is a improvement over the weak translation invariance achieved by data-enhancing~\cite{xuan2021deep}.
\item We utilize a generative adversarial network (GAN)~\cite{NIPS2014_5ca3e9b1} to efficiently generate candidate monomer density fields at saddle points, eliminating the need for memorizing the training set. This saves
 enormous computer memory and computation time, and makes it possible to generate a wider range of saddle point, monomer density field predictors. The GAN results can be used as the final monomer density predictions, which effectively eliminates the gradient descent procedure. 
\end{enumerate}

The rest of the paper is organized as follows. In Section \ref{sec:2d_def}, we define the problem and the overall strategy. 
We introduce the CNN trained in Sobolev space and detail how strong rotation/translation invariance is enforced in Section \ref{sec:2d_cnn}. We devote Section \ref{sec:2d_gan} to introduce the new approach to predict saddle point, monomer density fields by generating initial density fields via a GAN and then selecting the final  prediction with the Sobolev space-trained CNN. In Section \ref{sec:2d_result}, we present numerical results for a two-dimensional AB diblock polymer system that demonstrate the efficiency of the proposed machine learning approach to produce accurate approximations of the 
effective Hamiltonian and saddle point, monomer density predictions. We \revision{also present} results on the effectiveness of the proposed ML method to obtain accurate heat maps of the effective Hamiltonian as a function of the cell size and further equip ML methods with efficient gradient-based searching algorithms for cell size optimization. 
 Finally, we provide some concluding remarks in Section~\ref{sec:2d_con}. 

\end{section}

\begin{section}{Problem Definition and General Strategy}\label{sec:2d_def}
To formulate the problem and illustrate the methodology, we use the concrete example of a 2D incompressible AB diblock copolymer melt\revision{~\cite{fredrickson2006equilibrium}}. \revision{There are some important factors to consider in this study of the system. These include:
\begin{itemize}
\item $\chi N$: This measures the strength of
segregation of the two components (A and B). It is calculated using the Flory parameter $\chi$ and the degree of polymerization (the number of monomer units) of the copolymer $N$.
\item Cell vectors $\vec{a}_1$ and $\vec{a}_2$: These are the adjacent sides of the cell parallelogram being studied. Their length is measured in units of the unperturbed radius of gyration $R_g$ (root-mean-square distance of the segments of the polymer chain from its center of mass.).
\item $f$: This is the volume fraction of component A.
\item $\rho$: This is the local average monomer density of blocks A, measured in units of the total monomer density $\rho_0$. Henceforth, we will refer
to it as the density field or simply the density, with the understanding that it is computed or
approximated in a vicinity of SCFT saddle points. 
\end{itemize}
}

As in the one-dimensional case~\cite{xuan2021deep}, we formulate two problems to solve:
\begin{itemize}
    \item \textbf{Problem 1}: Learn a surrogate map $(\chi N,\vec{a}_1,\vec{a}_2,f,\rho)\longmapsto \tilde{H}(\chi N,\vec{a}_1,\vec{a}_2,f,\rho)$, where $\tilde{H}$ is an accurate approximation of the field-theoretic intensive Hamiltonian $H$, the Helmholtz free-energy per chain at saddle points. Henceforth, we simply call $H$ the Hamiltonian. 
    \item \textbf{Problem 2}: For specific values $\chi {N}^*$, $\vec{a}_1^*$, $\vec{a}_2^*$, $f^*$, find accurately and efficiently
    a density field $\rho^*$ that minimizes $\tilde{H}$. 
\end{itemize}

The map $\tilde{H}$ in Problem 1 is trained to approximate $H$ in the vicinity of saddle points. This surrogate functional can thus be evaluated much more efficiently than through direct SCFT computation and provides an expedited way to screen for the density field candidates in Problem 2. 

As done in~\cite{xuan2021deep},  we recast  $H$ as 
 \begin{align}\label{h_represent}
   H(\chi N,\vec{a}_1,\vec{a}_2,f,\rho)=\chi N f-\frac{\chi N}{|\vec{a}_1\times \vec{a}_2|}\int \rho^2 dr+ R(\chi N,\vec{a}_1,\vec{a}_2,f,\rho), 
 \end{align}
by extracting the leading quadratic interaction term and focus on learning the remainder $R$\revision{. As the difference between $H$ and quadratic interaction terms,  $R$ contributes to polymer entropy but not enthalpy in (\ref{h_represent}), which includes a partition function for a single copolymer with a field acting on the A
block and the other field acting on the B block. This single chain partition function is evaluated by integrating the copolymer propagator satisfying the Fokker-Planck equations. Details are shown in Appendix of ~\cite{xuan2021deep}.} \revision{In (1), }  $\times$ is the cross product and thus the denominator of the second term is the \revision{cell's area}. With this decomposition of $H$ we directly evaluate its main part, which is cheap to compute, and approximate the expensive part $R$ using machine learning. 

As noted in the Introduction, increasing the spatial dimension introduces two new significant challenges: \revision{the signficantly increased size of the discretized density field} and the additional invariance requirement of $H$ under spatial rotation transformations of $\rho$. To overcome these challenges,  we upgrade the methodology introduced in \cite{xuan2021deep} as follows. To solve Problem 1, instead of training a fully connected, feedforward neural network in Sobolev space, we train a CNN in Sobolev space to learn the effective Hamiltonian. To implement the CNN on the density fields, we view these as images and use them as input for the CNN. We employ a particular representation of the cell parameters $\vec{a}_1$ and $\vec{a}_2$ to strongly enforce 
the rotation invariance requirement. Moreover, by taking advantage of CNN's local shift invariance, we  design a CNN architecture 
that preserves strongly global shift invariance (in the one-dimensional setting we only obtained weakly global shift invariance via data enhancing). To solve problem 2, we propose a new strategy: use a GAN to generate a poll of initial  density fields.  In the 1D approach\cite{xuan2021deep}, our initial density fields come from the training set. With the GAN, it is not necessary to memorize the training set and perform evaluations to select a suitable initial field candidate. 
This GAN feature dramatically saves memory and computational cost as the training set is both large and high dimensional. 
 In addition, the randomness element in the GAN increases the possibility to generate unseen initial density fields and hence it enhances the capability of the proposed methodology for the discovery of new polymer phases. We finally 
  select  the predicted density field from the pool as the one that yields the smallest gradient for the Sobolev space-trained CNN of Problem 1. \revision{It is important to note that we only use direct SCFT simulations to generate the data needed to train the proposed ML models}.
 
\end{section}

\begin{section}{Methodology: Approximation of $H$ by a Convolutional Neural Network}\label{sec:2d_cnn}
\revision{For the central problem of constructing a surrogate functional for $H$,  we employ the decomposition (\ref{h_represent}) 
and focus on learning the remainder $R$ from SCFT-generated data. That is, our proposed approximation of $H$ has the form
\begin{equation}\label{approx_H}
    \tilde{H}(x) =\chi N f-\frac{\chi N}{|\vec{a}_1\times \vec{a}_2|}\int \rho^2 dr+ NN(x), 
\end{equation}
where $x=(\chi N,\vec{a}_1,\vec{a}_2,f,\rho)$ and $NN(\cdot)$ is an approximation of $R$ learned by a CNN. It takes density fields as image inputs along with model parameters $\chi N,\vec{a}_1,\vec{a}_2,f$.  }

\revision{CNN's are neural networks designed to process multiple arrays like those used for digital color images (three 2D arrays, each corresponding to the red, green, and blue colors). Thus, the input of a CNN can be viewed as a volume, or a 3D array. 
They have been notoriously successful in a variety of computer vision tasks \cite{yamashita2018convolutional}.
CNN's exploit local correlations of data by using convolutional layer and pooling layers.
 In a convolutional layer, each of the input  array values is replaced by
a weighted sum of a few local neighbors by taking an inner product with a sliding, small ($3 \times 3$ or $5\times 5$) matrix
of weights (called a filter bank or kernel) performing in effect a discrete convolution. The
weights are learned during training using stochastic gradient descent method and a different kernel can be used per channel.  Figure \ref{fig:convolution} shows a schematics of how a convolutional layer works. }

\revision{A pooling layer is a downsampling operation that reduces the dimension of the feature map~\cite{yamashita2018convolutional}. Common pooling layers include the Max Pooling
and the Average Pooling, where the maximal value or the average value in each neighbourhood is taken, respectively. 
Other types of pooling are a Global Maximal Pooling  and a Global Average Pooling, where the corresponding global values are selected. }

\revision{Activation layers, for example a ReLU function, increase the non-linearity and the learning ability of the CNN. After several stages of convolution, pooling, and activation layers, we can send the feature map into fully connected layers to generate the final output. We provide the specific architecture of the proposed CNN in Section~\ref{subsec:inv}.}


\begin{figure}
    \centering
    \includegraphics[scale=0.5]{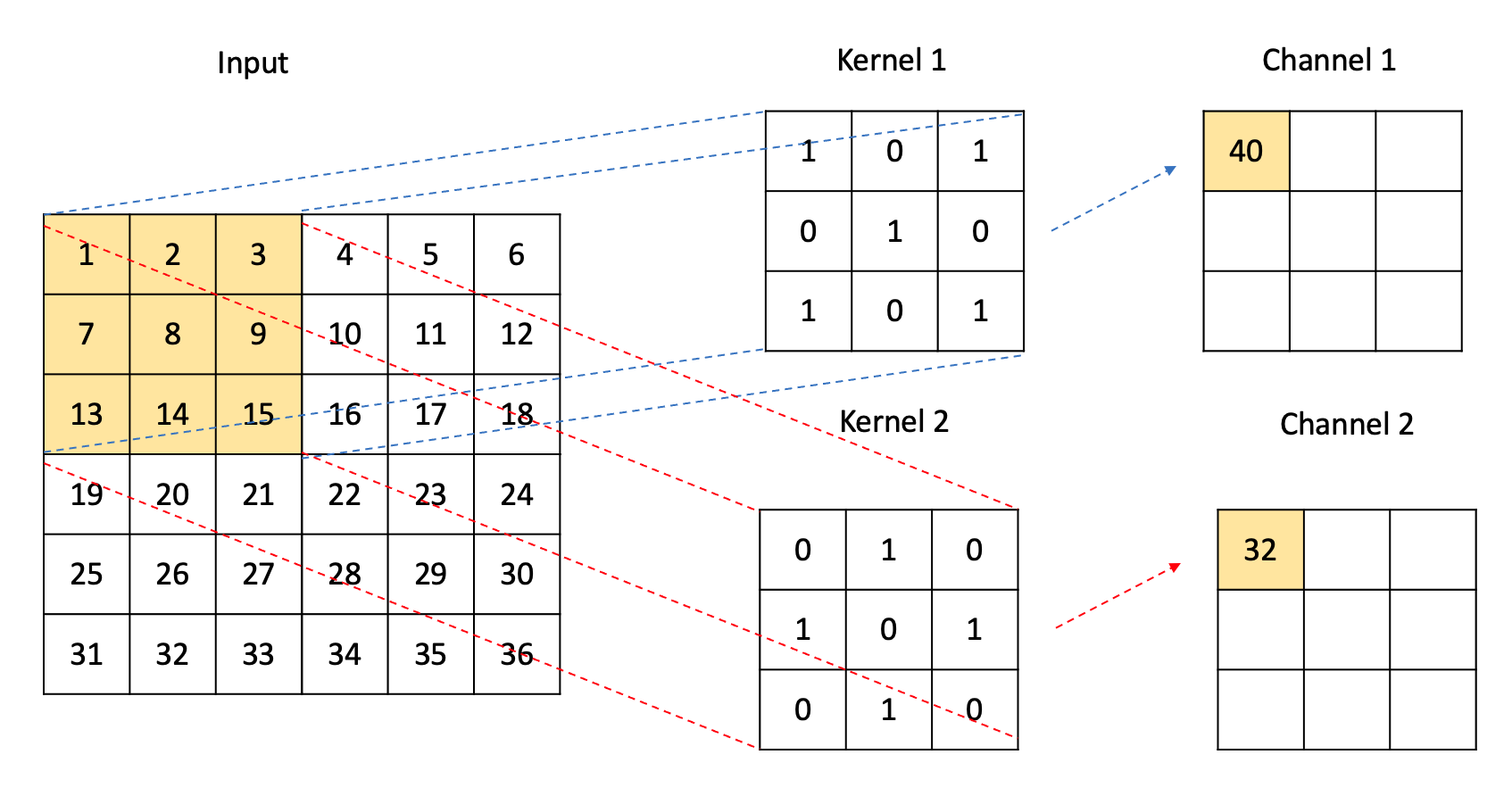}
    \caption{Illustration of a convolutional layer. \revision{Each of the input matrix values is replaced by a weighted sum of a few neighbors by taking an inner product with a sliding, small matrix of weights (called a filter bank or kernel) performing in effect a discrete convolution.
    The weights are learned during training. Channels can be used to account for extra dimensions in the input data. For example, a color image is composed of three 2D arrays (channels) corresponding to the pixel intensities in the red, green, and blue colors. Channels could be used to represent depth in 3D data. In this illustration, the weighted sum is done on element (2,2), with value 8,  of the input matrix.  }}
    \label{fig:convolution}
\end{figure}

\end{section}
\subsection{Rotation and Translation Invariance}\label{subsec:inv}

The central goal of Problem 1 is to 
produce an accurate approximation of $H$ which is invariant under both translations and rotations of the density field $\rho$. This invariance requirement  is illustrated in Fig.~\ref{fig:inv_rot_tran}, \revision{where the density fields in pictures (b) and (c)
result from translation and rotation, respectively, of the field in (a)}. All three pictures should 
give the  same value of the  Hamiltonian.
\begin{figure}
    \centering
    \includegraphics[scale=0.6]{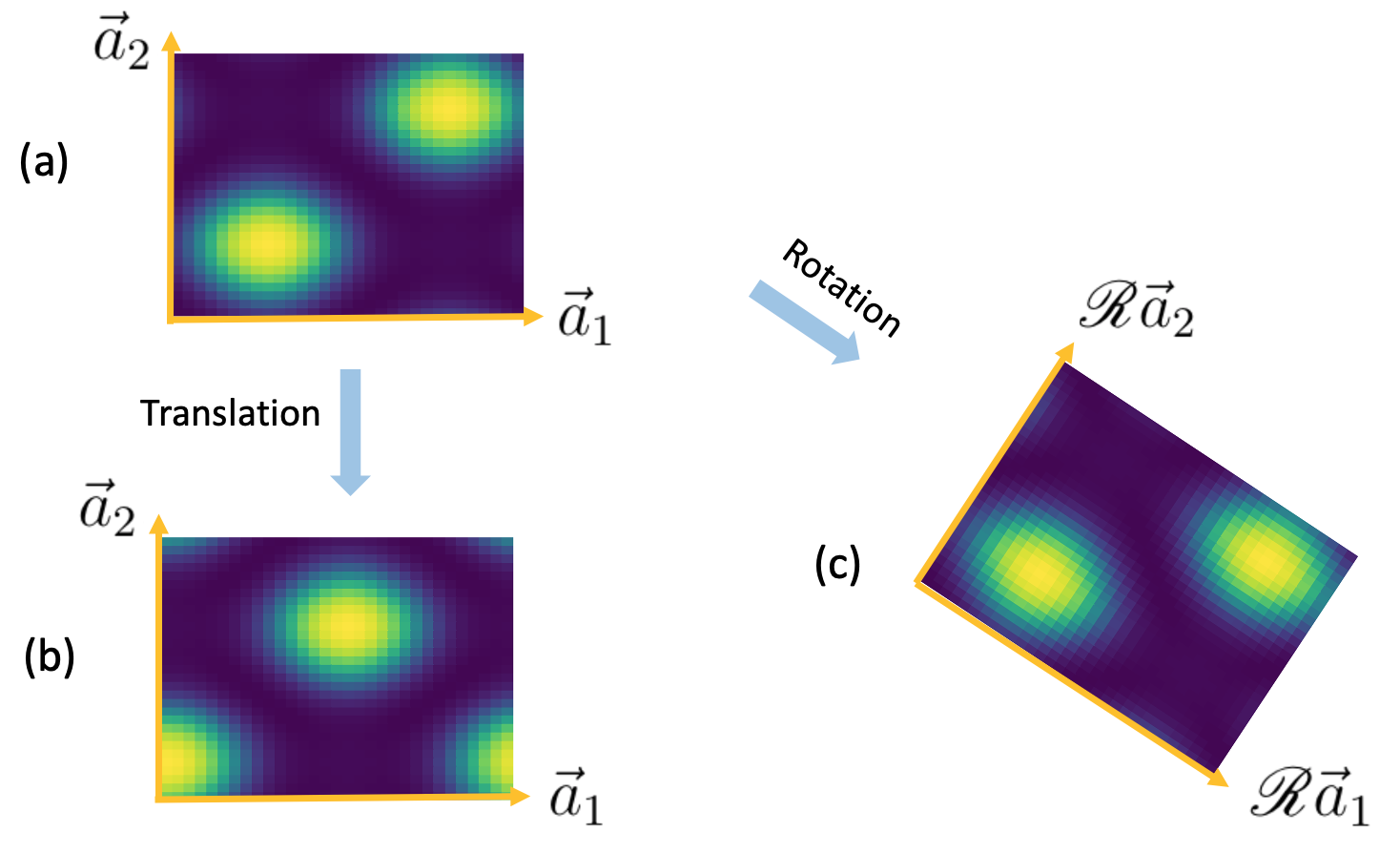}
    \caption{Rotation invariance and translation invariance. \revision{The field theoretic Hamiltonian $H$ must be invariant under spatial translations (2b) and rotations (2c) of the local, average monomer density field $\rho$ of blocks of type $A$ in a diblock copolymer melt (2a). The field $\rho$ is represented as a matrix corresponding to the values of this field at a uniform grid.} }
    \label{fig:inv_rot_tran}
\end{figure}
We opt to represent the 2D density field by the matrix of values of $\rho$ at a uniform (equi-spaced) grid. Mathematically, setting  $x=(\chi N,\vec{a}_1,\vec{a}_2,f,\rho)$, a  rotation $\mathscr{R}$ of the density field  could be represented by 
$$\mathscr{R}x=(\chi N,\mathscr{R}\vec{a}_1,\mathscr{R}\vec{a}_2,f,\mathscr{R}\rho)=(\chi N,\mathscr{R}\vec{a}_1,\mathscr{R}\vec{a}_2,f,\rho),$$
which means that such transformation  is equivalent to rotate just the two cell vectors, given the fixed arrangement of the density field matrix. Thus, the rotation invariance requirement can be succinctly expressed as 
\begin{equation}
\tilde{H}(\mathscr{R}x)=\tilde H(x).
\end{equation}
The enthalpic part of $\tilde{H}$ is rotation invariant  because $|\mathscr{R} \vec{a}_1\times \mathscr{R}\vec{a}_2|=|\vec{a}_1\times \vec{a}_2|$. Therefore, we seek a learner with the property
\begin{equation}
    NN(\mathscr{R}x)=NN(x).
\end{equation}
To achieve this we employ a rotation-invariant representation of ($\vec{a}_1$, $\vec{a}_2$). This is the natural representation $(l_1,l_2,\theta)$, where $\l_1=\|\vec{a}_1\|$, $\l_2=\|\vec{a}_2\|$ and $\theta$ is the angle between $\vec{a}_1$ and $\vec{a}_2$. Since $\|\vec{a}_i\|=\|\mathscr{R}\vec{a}_i\|$, for $i=1,2$  and $\theta(\vec{a}_1,\vec{a}_2)=\theta(\mathscr{R}\vec{a}_1,\mathscr{R}\vec{a}_2)$, using this representation for the input ensures the CNN output preserves strongly rotation invariance. Henceforth,  we set $x= (\chi N,l_1,l_2,\theta,f,\rho)$ and rewrite Eq~\ref{approx_H} as
\begin{align}\label{new_equation}
        \tilde{H}(x) =\chi N f-\frac{\chi N}{|l_1l_2\cos\theta|}\int \rho^2 dr+ NN(x).
\end{align}

Mathematically, the requirement of translation invariance can be expressed as
\begin{equation}
\tilde{H}(\mathscr{T}x)=\tilde H(x),
\end{equation}
where $\mathscr{T}x=(\chi N,l_1,l_2,\theta,f,\mathscr{T}\rho)$ and  $\mathscr{T}\rho(i_0+i,j_0+j)=\mathscr{T}\rho(i_0,j_0)$  for some $(i,j)$ and all  $(i_0,j_0)$. To build a CNN with 
this global property, we use \revision{a special padding technique on the boundary of each layer shown in Figure \ref{fig:circular}}. In the convolutional layer and the pooling layer of neural networks, the element-wise product of the kernel and the input near the boundary  is not well-defined because the boundary element does not have neighbours outside the boundary. The most common solution is to use zero padding, where we assume that there are neighbours with values zero near the boundary to compute the convolution, as shown in Figure \ref{fig:zero}. However, convolution with zero padding does not produce translation invariance. Instead, we endow the CNN with circular padding, which is similar to periodic boundary conditions in PDEs, where we append the values from the opposite boundary of the cell to each boundary. As shown in Figure \ref{fig:circular}, the periodic boundary allows the output tensor to have the same elements with a shift. \revision{In CNN, ``stride" refers to the number of pixels/elements the filter or kernel is moved horizontally or vertically across the input image/matrix. To achieve global translation invariance, } the convolutional layer and pooling layer should have 
stride 1 for their operators to commute with the translation operator. 

 We denote a convolution with size $k$, stride 1 and circular padding as $Conv_{k,1,p(k)}^c()$, where $p(k)$ is the size of padding which equals to $(k-1)/2$ when $k$ is odd and $k/2$ when $k$ is even (the size of padding is derived by how many rows/columns are required to match the kernel size when the boundary element is in the center of the kernel). Similarly, we denote the Average Pooling layer with size $k$, stride 1 and circular padding as $AvgP_{k,1,p(k)}^c()$, the Maximal Pooling layer with size $k$, stride 1 and circular padding as $MaxP_{k,1,p(k)}^c()$, the Global Average layer as $GlA()$, the Global Maximal layer as $GlM()$, the element-wise activation function layer as $AcF()$, (e.g. Relu as $Relu()$, Sigmoid function as $Sigmoid()$, etc.), and the fully connected layers as $FC()$.

\begin{figure}[!htb]
\begin{subfigure}{1.1\textwidth}
  \centering
 \includegraphics[width=.9\linewidth]{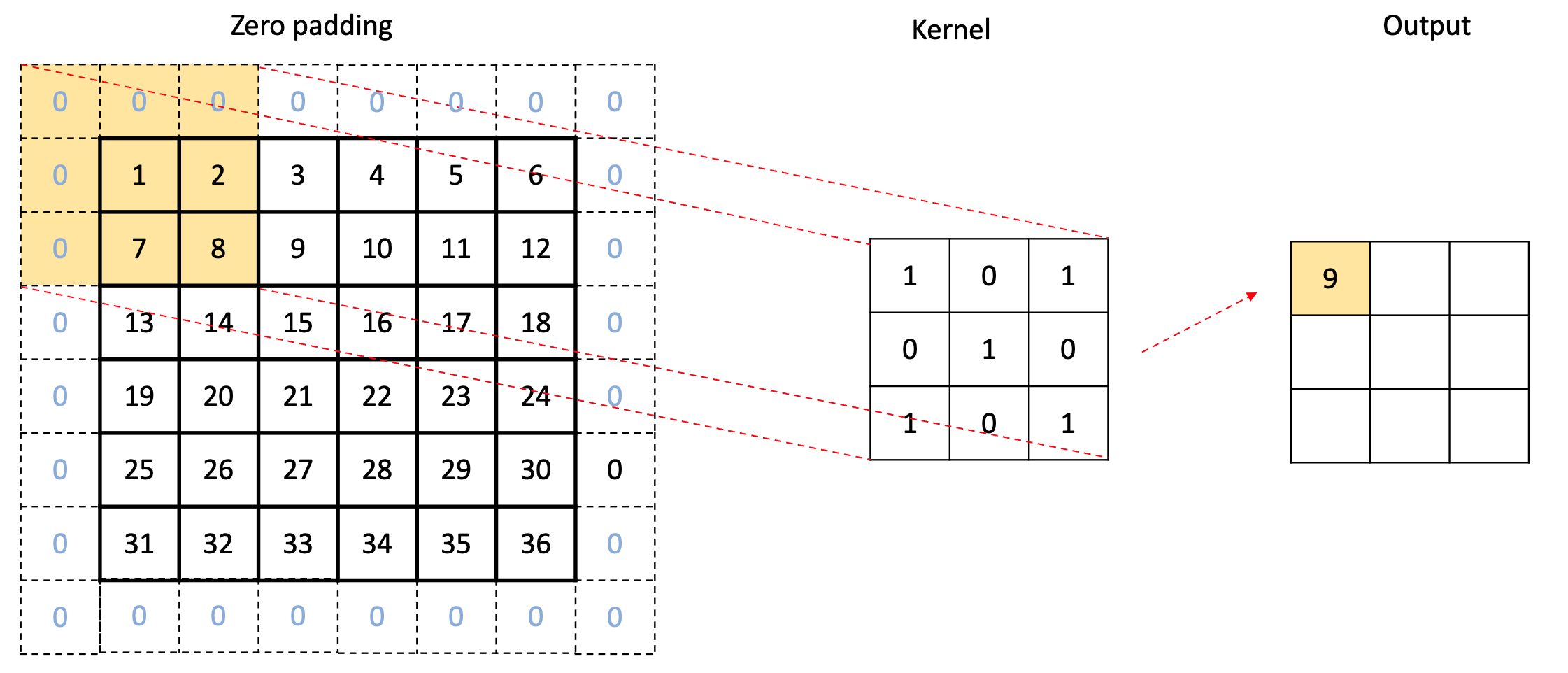}                                                                               
  \caption{Zero padding}
  \label{fig:zero}
\end{subfigure}

\begin{subfigure}{1.1\textwidth}
  \centering
  \includegraphics[width=.9\linewidth]{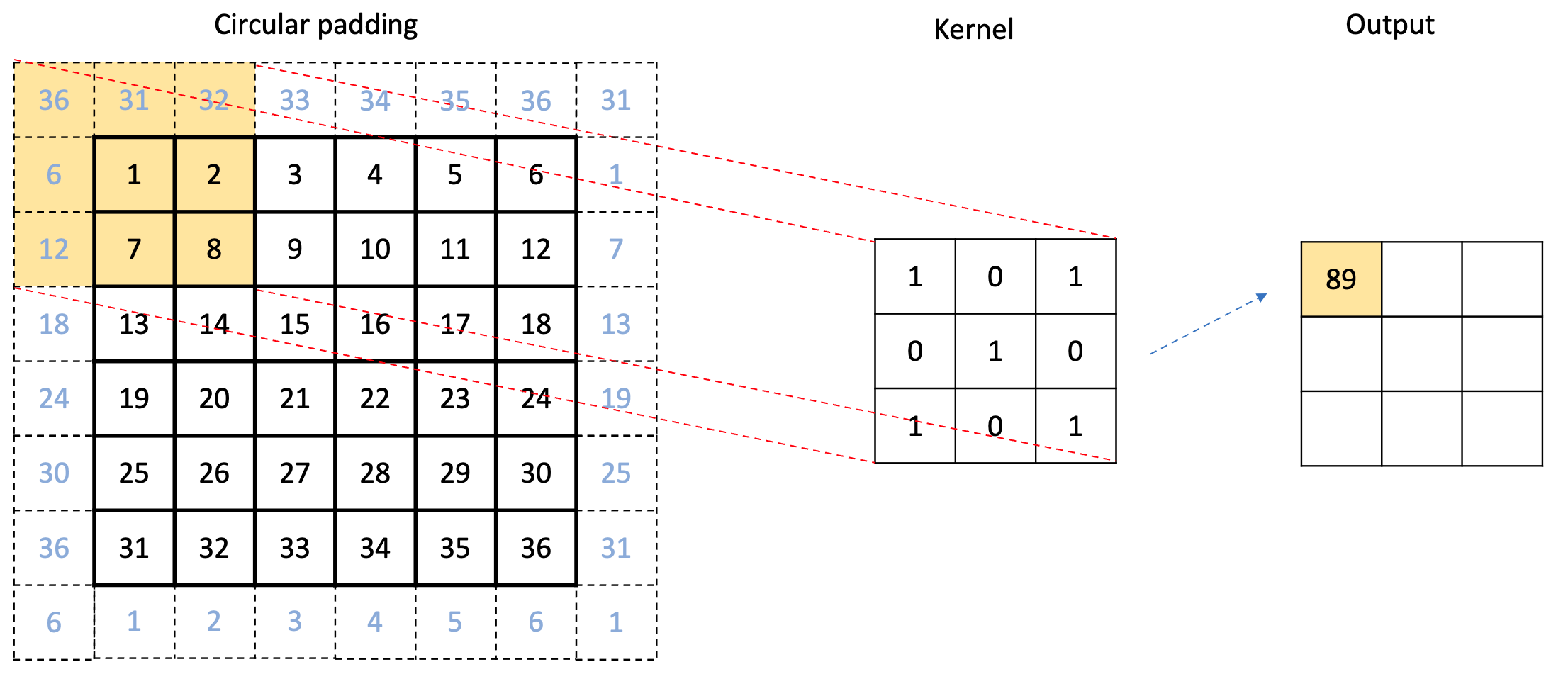}  
  \caption{Circular padding}
  \label{fig:circular}
\end{subfigure}
\caption{\revision{Padding in a convolutional layer. The input matrix is enlarged with the entries needed to perform the weighted average on the values at bottom and top rows and at the left-most and right-most columns. This is called padding. In a convolutional neural network, zero padding (a) is the standard. Here, we propose circular padding (b), which consists of wrapping around the data values, left and right and up and down, effectively enforcing periodic boundary conditions horizontally and vertically.}}
\label{paddings}
\end{figure}

\begin{lemma}\label{lemma:blocks}
$Conv_{k,1,p(k)}^c()$, $AvgP_{k,1,p(k)}^c()$, $MaxP_{k,1,p(k)}^c()$ and $AcF()$ all commute with the translation $\mathscr{T}$. $GlM()$ and $GlA()$ are invariant under translation.
\end{lemma}

\begin{proof}
Suppose the input tensor is $\revision{X}$ and a spatial translation of $\revision{X}$ is $\mathscr{T} \revision{X}$. Then, after the convolutional layer with stride 1 and circular padding,  $Conv_{k,1,p(k)}^c(\mathscr{T}\revision{X})$ only differs from $Conv_{k,1,p(k)}^c(\revision{X})$ by a shift because they have the same set of neighbourhood batches and only the starting batch changes. Thus,  $Conv_{k,1,p(k)}^c(\mathscr{T}\revision{X})=\mathscr{T} Conv_{k,1,p(k)}^c(\revision{X})$. Similarly,
$MaxP_{k,1,p(k)}^c(\mathscr{T}\revision{X})=\mathscr{T}MaxP_{k,1,p(k)}^c(\revision{X})$, and $AvgP_{k,1,p(k)}^c(\mathscr{T}\revision{X})=\mathscr{T}AvgP_{k,1,p(k)}^c(\revision{X})$. Now, since $AcF$ is an element-wise operation, $AcF(\mathscr{T}\revision{X})=\mathscr{T}AcF(\revision{X})$. Finally, it follows trivially that $GlM(\mathscr{T}\revision{X})=GlM(\revision{X})$ and $GlA(\mathscr{T}\revision{X})=GlA(\revision{X})$.
\end{proof}
We design the CNN with the following architecture: 
\begin{align}\label{eq:architecture}
     \rho \rightarrow \underbrace{Conv_{5,1,2}^c() \rightarrow Relu()}_{\times 5} \rightarrow GlA()\nonumber 
    \rightarrow \text{concat with }(\chi N,l_1,l_2,\theta,f)\rightarrow \underbrace{ FC()\rightarrow Relu()}_{\times 3}\rightarrow FC(),
\end{align}
where ``concat with $(\chi N,l_1,l_2,\theta,f)$" has the following meaning.
Following the application of $GlA()$ or $GlM()$, the output consists of multiple channels, each represented by a single number (either the global maximum or global average). These can be denoted as $(c_1, c_2, ... c_m)$, where $m$ is the total number of channels. We then concatenate the vector of all the channels with the parameters to get $(c_1,c_2,...c_m,\chi N, l_1,l_2,\theta,f)$, and put this into several fully connected layers to make the final predictions. Hyperparameters like the number of channels in each convolutional layer and the number of neurons in each fully connected layer are in Appendix \ref{appendix:hyperparameter}.

\begin{theorem}
The convolutional neural network (\ref{eq:architecture}) has shift invariance with respect to $\rho$.
\end{theorem}
\begin{proof}
From Lemma \ref{lemma:blocks}, 
\begin{align}
    GlA \circ (AcF \circ Conv_{k,1,p(k)}^c)^n (\mathscr{T} \rho)&= GlA \circ (\mathscr{T}\circ(AcF \circ Conv_{k,1,p(k)}^c)^n ( \rho))\\
    &=GlA \circ (AcF \circ Conv_{k,1,p(k)}^c)^n ( \rho),
\end{align}
 for any $n$, where $\circ$ stands for composition of operators. Since each channel outputs only one single number in the end, the fully connected layers preserves the translation invariance. Our case is a special case with $AcF=Relu$, $k=5$ and $n=5$, which inherits the translation invariance.
\end{proof}

Note that we could also add pooling layers with circular padding into the neural network, e.g. $AvgP_{k,1,p(k)}^c()$, $MaxP_{k,1,p(k)}^c()$, and the network would
 still preserve  translation invariance according to Lemma \ref{lemma:blocks}. Since the architecture   (\ref{eq:architecture}) without the pooling layers after convolutional layers already works very well we did not include the pooling layers for computation efficiency. 

Based on  Eq.~(\ref{approx_H}) and the considerations above for $NN(\cdot)$  we have
constructed an approximation $\tilde{H}$ to the Hamiltonian with global translation and rotation invariance, i.e.
\begin{align*}
    \tilde{H}(\mathscr{R}x)=\tilde H(x),\\
    \tilde{H}(\mathscr{T}x)=\tilde H(x).
\end{align*}

\begin{subsection}{Convolutional Neural Network in Sobolev Space}
We propose to train in Sobolev space the rotation and translation invariant CNN described in~\ref{subsec:inv} 
to fully solve Problem 1 and thus obtain a powerful tool for the solution of  Problem 2.

As in the one-dimensional case, the training data consists only of saddle points, at which the Hamiltonian has vanishing gradients. We encode this valuable information into the CNN's objective  function (also called cost or loss function) to train an approximation $\tilde{H}$ that matches both the Hamiltonian and its gradient. This approach has several advantages over the more common $L^2$ training. For instance, it makes the prediction more accurate in the vicinity of saddle points (by imitating the gradient behaviour of the true system) and it ensures the approximation has vanishing gradients on saddle points which is particularly significant for the predicted field in Problem 2. 

Considering the vanishing gradient at saddle points, we select CNN's parameters to minimize the objective function 
\begin{equation}\label{2d:cost}
        C(\alpha)=\sum_{i=1}^{N_T}\left(\tilde{H}(x_i)-H_i\right)^2 +\beta\sum_{i=1}^{N_T}\left\|\nabla_{\rho} \tilde{H}(x_i)\right\|^2,
\end{equation}
where \revision{$N_T$ is the size of the training set, } $\beta$ controls the importance of the penalty terms on gradients and $\alpha$ represents parameters in the CNN. It favours vanishing gradients at saddle points and also serves as a regularization for smoothness.

For each data point $(x_i,H_i)$, according to  Eqs. (\ref{h_represent}) and (\ref{approx_H}), $NN(x_i)$ should match 
\begin{equation}\label{remainder_formula}
     R_i = H_i -\chi N_if_i+\frac{\chi{N_i}}{|l_{1i}l_{2i}\cos{\theta_i}|}\int \rho_i^2 dr.
\end{equation}
Based on the spatial discretization, 
 \begin{align}
   \nabla_{\rho} \tilde{H}(x_i)&=-\frac{2\chi{N_i}}{l_{1i}l_{2i}} \, \Delta l_{1i} \, \Delta l_{2i} \, \rho_i +\nabla_{\rho}NN(x_i), \label{eq:2d_gradrho}
\end{align}
where $\Delta l_1$ and $\Delta l_2$ are the spatial mesh size of the $l_1$ direction and $l_2$ direction, respectively. $\cos \theta_i$ in (\ref{remainder_formula}) is cancelled with that in the area terms of small parallelograms in the discretized integral. Then we rewrite the cost 
function (\ref{2d:cost}) as 
\begin{equation}\label{2d:costnn}
  C(\alpha)=\sum_{i=1}^{N_T}\left(NN(x_i)-R_i\right)^2 +\beta\sum_{i=1}^{N_T}\left\|\nabla_{\rho} NN(x_i)-\frac{2\chi{N_i}}{l_{1i}l_{2i}} \, \Delta l_{1i} \, \Delta l_{2i} \, \rho_i\right\|^2,
  \end{equation}
where the  norm in the second term of (\ref{2d:costnn}) is the Frobenius norm of matrices.  $C(\alpha)$ is in the form of the (squared) Sobolev norm and we aim to train a map that matches both the functional and its gradients by minimizing this objective function. At the same time, with the CNN architecture (\ref{eq:architecture}), the approximation has rotation and translation invariance. Putting these two elements together, $\tilde{H}$ provides the desired solution to 
Problem 1 and the simultaneous accurate approximation of gradients near saddle points
will prove useful in the solution of Problem 2, i.e. the ultimate density predictions.

\end{subsection}

\begin{section}{Methodology: Density Field Prediction by Generative Adversarial Networks}\label{sec:2d_gan}
 In this section, we
focus on the solution to Problem 2 by employing a generative adversarial network (GAN) and the Sobolev space-trained CNN of Problem 1. 

\subsection{GAN, cGAN and DCGAN}
Our methodology to predict saddle point densities is built by merging  GAN, cGAN (Conditional GAN) and DCGAN (Deep Convolutional GAN). We briefly introduce next 
 \revision{these} three architectures and the motivation to apply them for density predictions. 

A GAN \cite{NIPS2014_5ca3e9b1} is a \revision{neural network based on a two-player minimax game model: a generator (G), which is trained to generate from random noise to the images similar to the images in the training set and a discriminator (D), which is trained to estimate the probability that the input image is from the training set rather than from G}. After several rounds of training, G is expected to generate images which could fool D and D is expected to tell the difference from the true image and the generated (``fake") image. \revision{Thus, after a successful training}, G takes random noise as input and generates images that imitate the training set. For instance, if the training set are pictures of human faces, G is expected to generate pictures of human faces because it learned how to extract the latent features of these type of images. \revision{G and D play a two-player minimax game with  a value function $V(G,D)$:}
\begin{equation}\label{loss_vanila_gan}
\min_G\max_D V(D,G)=E_{x \sim p_{data}(x)}[\log D(x)]+E_{z \sim p_z(z)}[\log (1-D(G(z)))],
\end{equation}
where $D(x)$ is D's estimate of the probability that data $x$ is real (i.e. from the training set), 
$G(z)$ is G's output given noise z, and $E_{x \sim p_{data}(x)}$ and $E_{z \sim p_z(z)}$ are the expected values over all real data and over all the random inputs to G, respectively. Based on this loss function, D is trained to maximize the probability that $D(x)=1$ for all the $x$ in the training set and the probability that $D(G(z))=0$ for all the images generated by G. G is trained to minimize the probability that D identifies the fake images and fools D by driving $D(G(z))$ to 1. 

Conditional GAN (cGAN) was proposed~\cite{mirza2014conditional} to send labels to G for the generation of images corresponding to the label. The architecture and training of a cGAN is similar to those of GAN, except that G has an extra input, the labels. For example, for the cGAN trained on  MNIST
, which is a dataset containing the images of 60K handwritten digits,  G has two types of inputs: random noise and labels of the digits. After training, G is expected to generate the digits corresponding to the true label.

DCGAN~\cite{radford2015unsupervised} is an extended version of GAN, with an architecture where D is mainly made up of convolutional layers and G is mostly composed of  transposed convolutional layers. A transposed convolutional layers is an upsampling layer, whose output has larger dimension size than the input. A detailed introduction of transposed convolutional layer \revision{can be found in \cite{dumoulin2016guide}.}
Since the transposed convolutional layers have the ability to expand dimensions, we can view the noise vector of dimension $n$ as an image with $n$ channels and size $1\times 1$ in each channel. Then by repeating transposed convolutions, each channel can be expanded to the image size, e.g. $32\times 32$. With the transposed convolutional layers in G and convolutional layers in the D, DCGAN has the potential to  generate higher quality images.

\subsection{ScftGAN}\label{sec:Scftgan}
As described above,  GAN-based models attempt to learn a map from latent feature space to a manifold of images with similar styles or features. \revision{They generate data from this manifold using a random noise with fixed dimension as an input from latent feature space}. In  Problem 2, \revision{the prediction of a saddle point density field}, we can regard all the density fields as points on a high dimensional manifold $\mathcal{M}_s$. 
\revision{We propose to employ a cGAN to learn the map from a latent feature space to a manifold approximating $\mathcal{M}_s$ and use G to generate density fields from it. 
In order to predict density fields for a specific combination of parameters, ($\chi {N}^*$, $l_1^*$, $\l_2^*$, $\theta^*$, $f^*$), we  also pass these parameters to G as input. An illustration of cGAN prediction of hand-written digits and of density fields by cGAN is shown in Fig.~\ref{fig:anology}. After training with the structure of $\mathcal{M}_s$ together with the relationship between parameters and density fields, G is expected to generate density fields for the given parameters from random noise.}

\begin{figure}
    \centering
    \includegraphics[width=.95\linewidth]{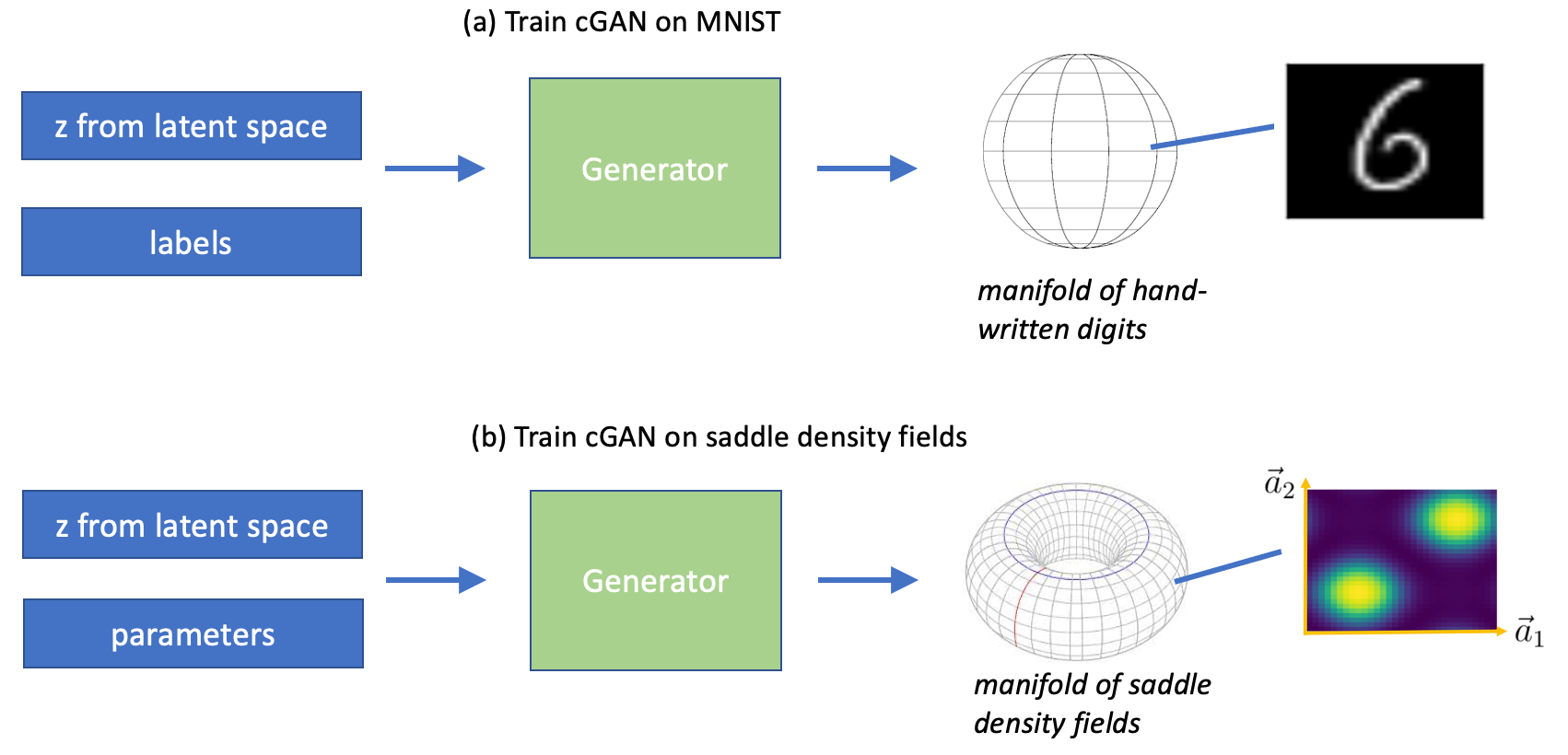}
    \caption{\revision{Illustration of conditional generative adversarial network (cGAN) for generation of (a) handwritten digits and (b) density fields. After successful training, the cGAN 
    receives as input a random noise $z$ and a label, for example label 6 in (a) or the parameters ($\chi {N}^*$, $l_1^*$, $\l_2^*$, $\theta^*$, $f^*$) in (b), and generates data similar to the real data with that label. Effectively, a trained cGAN learns a map from a compressed set of data features (latent space) to a manifold of data which closely resembles the real data.}}
    \label{fig:anology}
\end{figure}

The new GAN architecture  we propose for density field prediction, which we henceforth call ScftGAN, is shown in Fig.~\ref{fig:scftgan}. G generates ``fake" density fields from random noise and the parameters ($\chi {N}^*$, $l_1^*$, $\l_2^*$, $\theta^*$, $f^*$). The fake density fields together with the parameters (shown as Information in Fig.~\ref{fig:scftgan}) are sent to D. These instances are labeled ``fake" to train D. The true density fields are also sent to D, labeled as ``true" during  the training process. With the education from the true and fake instances, D is trained to accurately classify them. In turn, the results of D are  utilized to update G in a stochastic gradient descent method, by maximizing the probability that D identifies as ``true" the fake instances (generated by G). 
 To train D with the relationship between parameters and the corresponding density field, we also send the true density field with  shuffled Information (parameters), i.e. assigning a random Information (parameters) from the pool to the density fields, and penalizing the probability that D classify these as true pairs. This is to avoid the risk that D makes a decision without considering the parameters. The knowledge that only density fields corresponding to the input parameters should be generated, is enhanced in training.  Furthermore, given that both the  Hamiltonian CNN predictor and the ScftGAN are expected to predict accurately around a vicinity of the training set, we anticipate that the output of ScftGAN will produce a value of  $\tilde{H}$ close to the true Hamiltonian after sending it to the CNN in case the  ScftGAN predictions are far from the training domain. 
 
 The loss function for our proposed ScftGAN is 
\begin{align}\label{scftgan_loss}
\min_G\max_D V(D,G)=2*E_{\rho \sim p_{data}(\rho)}[\log D(\rho |y)]+E_{z \sim p_z(z)}[\log(1-D(G(z|y)|y))]\nonumber\\
+E_{\rho \sim p_{data}(\rho)}[\log (1 - D(\rho|S(y)))]+\lambda E_{z\sim p_z(z)}[H(y)-\tilde{H}(G(z|y))]^2,
\end{align}
where $y$ is the tuple of parameters ($\chi {N}$, $l_1$, $\l_2$, $\theta$, $f$), $S(y)$ means a random shuffle of all the Information (parameters) in the data set (a random permutation of Information among data points), and 
$\lambda$ is a hyperparameter controlling how far away is the output  from the vicinity of the training set.
The first term is doubled to balance the ratio of true labels and false labels (both the second term and third term correspond to false labels).
Compared with the loss function (\ref{loss_vanila_gan}) of the classical GAN, parameter information is sent to both G and D, and the additional term enforces D to output``0" for the input of wrong pairings of parameters and density fields, so that D can learn the correct relationship between parameters and  density fields.
\begin{figure}
    \centering
    \includegraphics[scale=0.7]{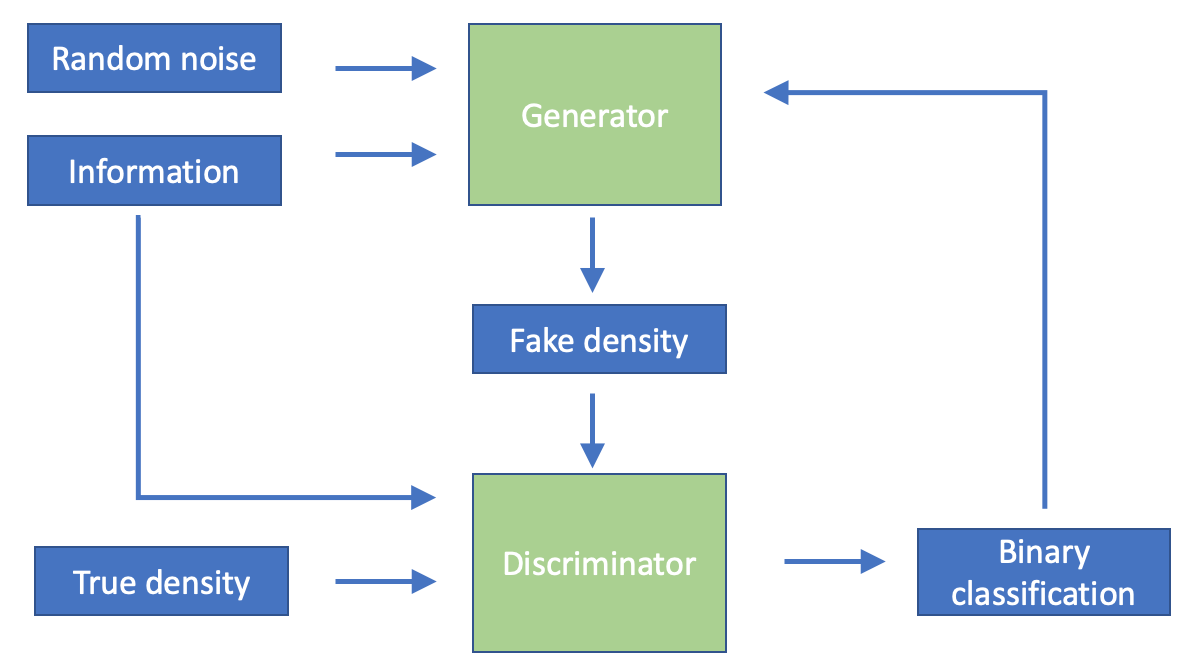}
    \caption{ScftGAN, \revision{a conditional generative adversarial network (cGAN) for polymer self consistent field theory (SCFT).  The information block represents the diblock parameters ($\chi {N}$, $l_1$, $\l_2$, $\theta$, $f$). During training, the generator generates ``fake" density fields from random noise and the given parameters. The generated field and the parameters are passed to the discriminator, which is also fed with the true density field. With this information, the discriminator is trained to accurately classify fake or true. In turn, the results of the discriminator (a binary classification) are  utilized to update the generator by maximizing the probability of fooling the discriminator. After successful training, the generator produces density fields which closely resemble those originating from direct SCFT.} }
    \label{fig:scftgan}
\end{figure}
The architecture of G is :
\begin{align}\label{scftgan_generator}
    z&\rightarrow \text{concat with }y \rightarrow ConvT_{4,1,0}() \rightarrow Bnorm() \rightarrow Relu() \nonumber\\
    &\rightarrow \text{concat with }y 
   \underbrace{\rightarrow ConvT_{4,2,1}() \rightarrow Bnorm() \rightarrow Relu()}_{\times 2} \nonumber\\
    &\rightarrow \text{concat with }y \rightarrow ConvT_{4,2,1}() \rightarrow Bnorm() \rightarrow Relu()\nonumber\\
    &\rightarrow \text{concat with }y 
   \underbrace{\rightarrow ConvT_{5,1,2}() \rightarrow Bnorm() \rightarrow Relu()}_{\times 6}\nonumber\\
    &\rightarrow ConvT_{5,1,2}()\rightarrow Sigmoid(), 
\end{align}
\revision{where $z$ is random noise (as Appendix \ref{appendix:hyperparameter} shows, we set the dimension of noise 16)}, $ConvT_{m,n,k}()$ stands for transposed convolution of size m, stride n, and zero padding of size k. Transposed convolution is an upsampling operation, for instance, $ConvT_{4,1,0}()$ with kernel size 4, stride 1 and no padding will expand a image of $1\times1$ (the initial size of each channel) to $4 \times 4$. Then each transposed convolution $ConvT_{4,1,0}$ will double the image size because stride is 2 and the image size of each channel  achieves $32 \times 32$ after 3 operations. The following transposed convolution $ConvT_{5,1,2}$ will not change the image size of each channel. $Bnorm()$ represents batch normalization, $LeakyRelu()$ is an activation function where input value less than zero is multiplied by a fixed scale factor, ``$\text{concat with }y$'' refers to the operation that treats each parameter in $y$ as a constant channel and append the parameter channels to the input tensor. We frequently repeat inputting the parameters to the neural networks to stress the importance of parameters in density fields prediction.

The architecture of D is: 
\begin{align}\label{scft_discriminator}
    \rho &\rightarrow \text{concat with }y \rightarrow Conv_{3,1,1}() \rightarrow LeakyRelu()\nonumber\\
       &\underbrace{\rightarrow Conv_{4,2,1}() \rightarrow Bnorm() \rightarrow LeakyRelu()}_{\times 3} \nonumber\\
       &\rightarrow Conv_{4,1,0}() \rightarrow Sigmoid().
\end{align}
The output is a number between 0 and 1 which measures the probability that the input density field $\rho$ is the field of the corresponding parameters  $y$.

The hyperparameters of $G$ and $D$, like learning rate, number of channels in convolution/transposed convolution layers, etc., are listed in Appendix \ref{appendix:hyperparameter}.

After training, given specific parameters ($\chi {N}^*$, $l_1^*$, $l_2^*$, $\theta^*$, $f^*$), GAN can be used to generate a batch of candidates for density fields by taking various inputs of random noise, as we show in Section \ref{sec:2d_result}. Then, as we detail next,  we can employ the Sobolev space-trained CNN  to select an optimal density field, with optional fine-tuning, and predict the corresponding  Hamiltonian.

\begin{subsection}{Density Field Prediction}

Since we perform Sobolev space-training for the approximation $\tilde{H}$ of the   Hamiltonian, the true  density fields will yield small or vanishing gradients of $\tilde{H}$. Therefore, we plug all the density fields generated by ScftGAN into $\tilde{H}$ and compute  $\nabla_{\rho}\tilde{H}$. We select the density field in the pool with the smallest $\|\nabla_{\rho} \tilde{H}\|$ as the density field prediction. The proposed procedure of density prediction is illustrated in Fig.~\ref{fig:saddle_logic}.
\begin{figure}[!htp]
    \centering
    \includegraphics[width=.9\linewidth]{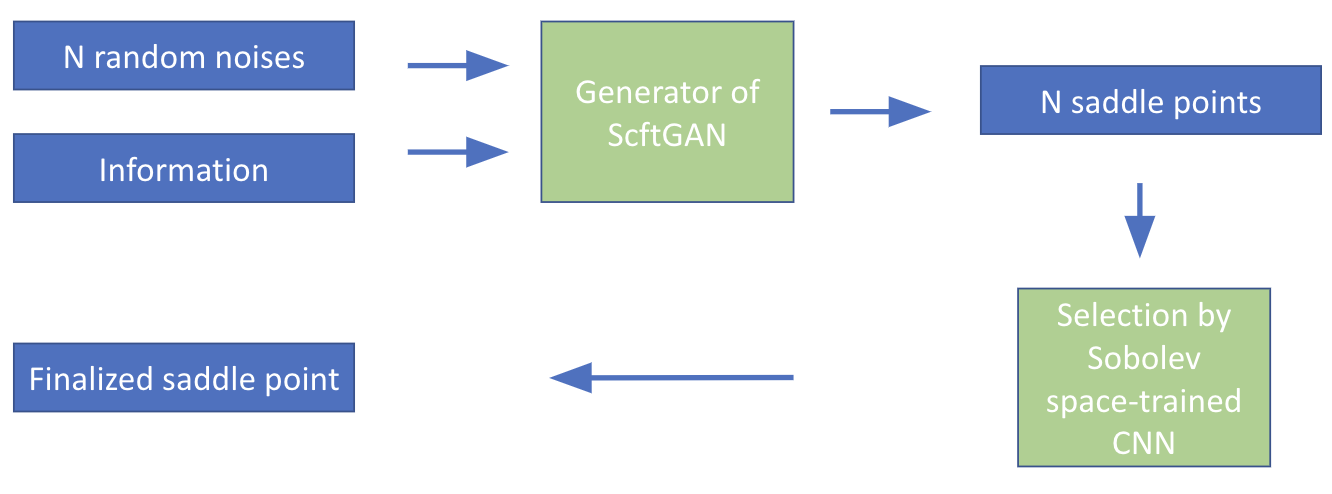}
    \caption{Procedure for density field prediction. \revision{A batch of $N$ density fields (the $N$ saddle points block) is generated with the ScftGAN from an input of $N$ random noises and the given diblock parameters (information block). Then, the Sobolev space-trained convolutional network is employed to expeditiously evaluate $\nabla_\rho \tilde{H}$ for each field in the batch. The final prediction is the field with the smallest $\|\nabla_\rho \tilde{H}\|$ .} }
    \label{fig:saddle_logic}
\end{figure}  

It is important to contrast the current procedure with the one we employed in the 1D setting~\cite{xuan2021deep}.  In the latter we selected from the training set the density field with the smallest $\|\nabla_{\rho}\tilde{H}\|$ as an initial density field   and then use the Sobolev space-trained $\tilde{H}$ to improve this initial guess via gradient descent: 
\begin{equation}\label{GD}
    \rho^{n+1}=\rho^{n}-\epsilon \nabla_{\rho} \tilde{H}(x^{n}).
\end{equation} 
There are several advantages of the new GAN-based approach over the 1D procedure for obtaining a density prediction.  
 First, there is now no need to memorize the training set and this saves an enormous amount of computer memory.   
 Second, the strong dependence of performance on the training set is relaxed so that the algorithm is more robust. Third, the \revision{ScftGAN} could generate multiple candidates given different values of noise, which broadens the diversity of the candidate pool and consequently increases the probability to discover new polymer phases. Fourth, we could use the result of ScftGAN as an initial guess for the gradient descent ({GD}).
 However, as the results in Section~\ref{sec:2d_result} show,  ScftGAN itself already generates excellent results without that fine-tuning step.
 
 The memory and time complexity comparison of the training set-based method with fine-tuning in~\cite{xuan2021deep} and the new ScftGAN method are shown in Table~\ref{compare_complexity}. The training set-based method has $O(N_T)$ memory cost ($N_T$ is the number of points in the training set) 
 and a time cost $O(N_T)+O(K)$ for the evaluation and selection of initial density fields and for $K$ steps in fine-tuning by the gradient descent method.  In contrast,  ScftGAN-based method only needs to store the constant number of parameters, which are independent of the size of training set. In terms of time cost, ScftGAN-based method only needs to evaluate and select candidates in the pool which is a small constant compared to $N_T$.
\begin{table}[htp!]
\centering
\begin{tabular}{lcc}
\hline
                 & Training set-based selection+ fine-tunning & ScftGAN \\ \hline
Memory cost      & $O(N_T)$                             & $O(1)$  \\ \hline
Time cost & $O(N_T)+O(K)$                        & $O(1)$  \\ \hline
\end{tabular}
\caption{Memory and computational cost. $N_T$ is the size of training set and $K$ is the number of gradient descent steps.}
\label{compare_complexity}
\end{table}

\subsection{Algorithm for the Overall ML Methodology}
We summarize in Algorithm~1 the proposed ML methodology for  Hamiltonian and saddle point density predictions. \revision{In step 1, we generate the dataset  (training set, validation set and test set) by solving SCFT equations.} In step \revision{2}, we learn a surrogate  Hamiltonian $\tilde{H}$ as the sum of the enthalpic part and a CNN remainder trained in Sobolev space, which strong, global translation and rotation invariance.  In step \revision{3}, we train a ScftGAN, consistent with $\tilde{H}$, whose generator G has the capability to predict  density fields from parameters $\chi {N}^*$, $\vec{a}_1^*$, $\vec{a}_2^*$, $f^*$ and random noise. In step \revision{4}, we generate a pool of  density field candidates and from this pool we select the field $\rho^{ML}$  that produces the smallest $\|\nabla_{\rho} \tilde{H}\|$.
 In the last step, we plug $\rho^{ML}$ into $\tilde{H}$ to predict the corresponding value of the  Hamiltonian. 

\begin{algorithm}[!htp]\label{algo_2d}
\setstretch{1.25}
\SetAlgoLined
\textbf{\revision{1. Generate the dataset consisting of data tuples $(\chi N, \vec{a}_1, \vec{a}_2, \rho, H)$ by solving   for SCFT system.}}\\
\textbf{\revision{2}. Train the  Hamiltonian surrogate $\tilde H$:}\\
Train
\begin{align*}
    \tilde{H}(x) =\chi N f-\frac{\chi N}{|\vec{a}_1\times \vec{a}_2|}\int \rho^2 dr+ NN(x), 
\end{align*}
by finding $NN(x)$ that minimizes the loss function (\ref{2d:costnn}).\\
\textbf{\revision{3}. Train the  density field predictor ScftGAN:}\\
 Train ScftGAN by minimizing (\ref{scftgan_loss}) using $\tilde{H}$ in Step 1 and get the corresponding density generator $G$  (\ref{scftgan_generator}).\\
\textbf{\revision{4}. Prediction of density corresponding to $\chi {N}^*$, $\vec{a}_1^*$, $\vec{a}_2^*$, $f^*$:}\\
\For{\text{Random noise} $n_i$, $1\leq i \leq N$}
{Evaluate $\rho_i=G(\chi {N}^*, \vec{a}_1^*, \vec{a}_2^*, f^*, n_i)$ and $\nabla_{\rho}\tilde{H}(\chi {N}^*, \vec{a}_1^*, \vec{a}_2^*, f^*,\rho_i)$\\
}
$\rho^{ML}=\argmin_{\rho_i}\|\nabla_{\rho}\tilde{H}(\chi {N}^*, \vec{a}_1^*, \vec{a}_2^*, f^*, \rho_i)\|$.\\
 \textbf{\revision{5}. Prediction of the corresponding  Hamiltonian:}
 $$H^{ML}=\tilde{H}(\rho^{ML}).$$
 \caption{ML method to learn the  Hamiltonian and to obtain a  density field prediction.}
 \label{algorithm}
\end{algorithm}

The Sobolev-space trained Hamiltonian predictor $\tilde{H}$ and saddle density field predictor ScftGAN can work together as an end-to-end machine learning solution for downstream tasks, like the optimal cell size and shape optimization that people are generally interested in. The methods and experiments to the application are shown in Section  \ref{heatmap_screening} and Section \ref{efficient_heatmap}.

\end{subsection}

\end{section}

\begin{section}{Results}\label{sec:2d_result}
In this section, we present results of numerical experiments for a two-dimensional AB diblock system to validate the efficacy and efficiency of the proposed methods. The data are generated for the following, common range of parameters, $\chi N = 16$, $l_1 \in [3,5.5]$, $l_2 \in [3,5.5]$, $\theta \in [\pi / 2, 5\pi/6]$, $f \in [0.3, 0.5]$. We sample data points on equidistributed nodes in the given interval of each parameter by running a direct SCFT solver to compute the corresponding  density fields and the  Hamiltonian. The modified diffusion equations are solved using periodic boundary conditions and pseudo-spectral collocation in space with $32\times 32$ mesh points in 2D.  Auxiliary fields, initialized with smooth fields with a fixed number of periods, are relaxed to saddle-point configurations using the semi-implicit Seidel iteration~\cite{CF2004}. The cell lengths $l_1$, $l_2$ are specified rather than relaxed to the value that minimizes the  Hamiltonian. After building the dataset (24167 samples), we randomly pick $70\%$  of the data as training set, $15\%$  as the validation set, and $15\%$ as the test set. In the generated dataset there are hexagonal cylinders (HEX), lamellar phase (LAM) and square-packed cylinders. Our experimental results demonstrate the capability of the  Sobolev space-trained CNN to accurately approximate the  Hamiltonian for both phases, without having to train separately a learner for each phase.

\subsection{Approximation of the  Hamiltonian}
We summarize next numerical results on the accuracy of the Sobolev space-trained CNN to approximate the  Hamiltonian and its gradient. Table \ref{tab:cnn} lists the root mean square error for these quantities. The Sobolev space-trained CNN achieves 3 digits of accuracy on $H$ prediction and 2 digits of accuracy on its gradient.
\begin{table}[htp!]
\centering
\begin{tabular}{lcc}
\hline
             & Hamiltonian Error & Hamiltonian Gradient Error \\ \hline
Training set & 0.00601            & 0.01173                     \\ \hline
Test set     & 0.00605            & 0.01171                     \\ \hline
\end{tabular}
\caption{(Root Mean Square) Error of the CNN}
\label{tab:cnn}
\end{table}
As Fig.\ref{fig:2d_h_compare} shows, this excellent accuracy holds for the entire range of parameters. This figure was generated by  
randomly picking 25 data points form the test set and plotting both the predicted and the SCFT  Hamiltonian. 
Figure \ref{fig:2d_h_compare} shows that the CNN-predicted prediction is accurate throughout the entire range of parameters. 
\begin{figure}[!htp]
    \centering
    \includegraphics[width=.9\linewidth]{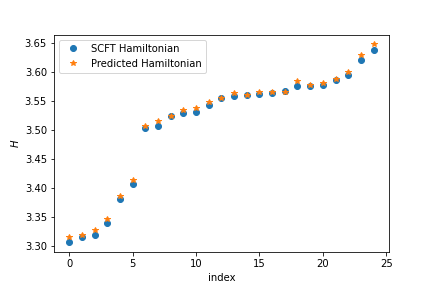}
    \caption{Comparison of \revision{the} SCFT Hamiltonian and its \revision{convolutional network} approximation for a random set of 25 data points. \revision{The data points are sorted by the SCFT Hamiltonian vaules.}}
    \label{fig:2d_h_compare}
\end{figure}

\revision{We also conduct a thorough study on the relationship of prediction accuracy and training set size. Detailed results and discussions are in Appendix \ref{train_set_size_analysis}.}

\subsection{Density Field Predictions}
The trained ScftGAN takes random noise and the parameter set ($\chi {N}^*$, $l_1^*$, $\l_2^*$, $\theta^*$, $f^*$) as input. The training process is illustrated in Fig.~\ref{gan_train}.
We choose three representative density fields, shown as rows,  for this figure. The first column corresponds to the (``true")  density computed directly by SCFT simulation. The other columns are the ScftGAN predicted density, with the same input noise, after different number of training epochs.  In the beginning (epoch 7), the ScftGAN predicts random density fields but the predictions significantly improve as the number of training epochs increases.

\begin{figure}[!htp]
    \centering
    \includegraphics[width=.9\linewidth]{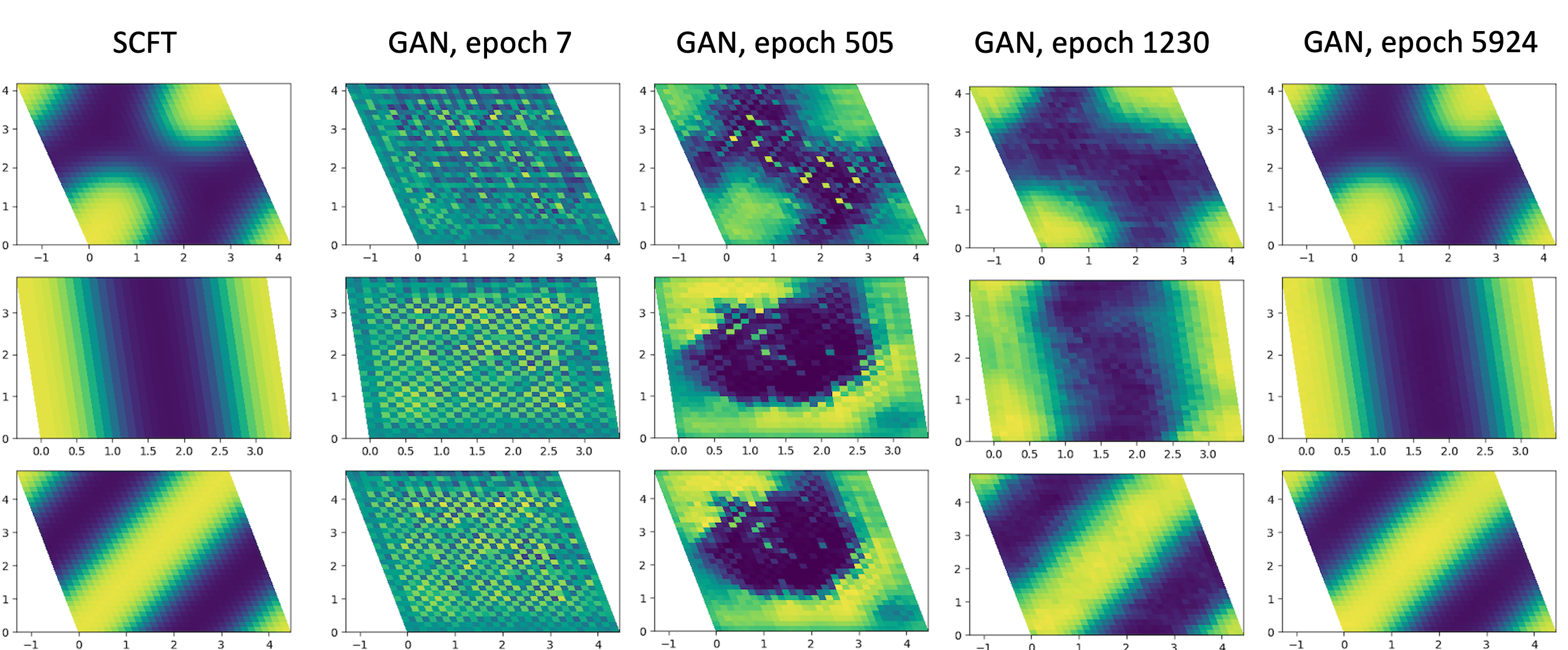}

\caption{Predictions of density fields on the validation set during training of the ScftGAN. Each row shows how the prediction changes for a representative case as the number of training epochs \revision{(proportional to stochastic gradient iterations)} increases. Column 1: the true (SCFT) saddle density field, Columns 2-5:  the density field predicted by the ScftGAN at different number of epochs. }
\label{gan_train}
\end{figure}

\end{section}

After training the ScftGAN we generate 5 density field candidates based on 5 different values of input random noise for each combination of parameters, ($\chi {N}^*,l_1^*$, $l_2^*$, $\theta^*$, $f^*$). From these fields,  we 
select the one that produces the smallest gradient of the Sobolev space-trained CNN approximation $\tilde{H}$. To demonstrate the efficacy of this procedure in the 
prediction of  density fields we take some representative cases from the test set. The results are summarized in Fig.~\ref{fig:test_evaluate}. The left column contains the  density fields computed directly by SCFT and the right column has the density fields predicted by the ScftGAN (finalized prediction selected by $\nabla \tilde{H}$). The true and the predicted fields look indistinguishable within the plotting resolution. 



\begin{figure}[!htbp]
    \centering
    \includegraphics[width=\linewidth]{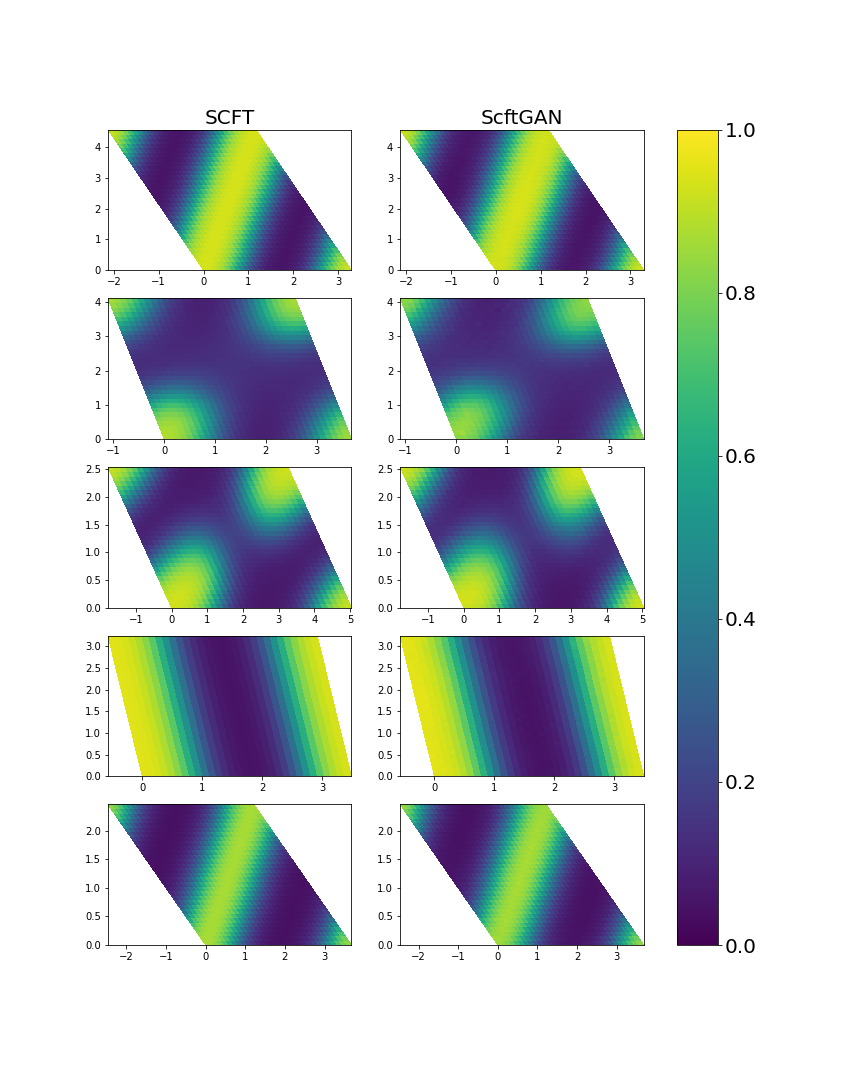}
    \caption{Comparison of the SCFT density field (left) and the corresponding \revision{machine learning prediction obtained with the combined ScftGAN and the Sobolev-space trained convolutional network} (right). Each row is a comparison of one representative case. The parameters $(l_1,l_2,\theta,f)$ for the 5 representative cases shown are as follows: row 1: $(3.4,5.2,2.01,0.48)$, row 2: $(3.8,4.4,1.83,0.3)$, row 3: $(5.2,3.2,2.18,0.36)$, row 4: $(3.60,3.40,1.74,0.50)$, row 5: $(3.80,3.60,2.36,0.38)$.}
    \label{fig:test_evaluate}
\end{figure}


\subsection{Cell size relaxation}\label{heatmap_screening}

An important SCFT downstream task is finding the optimal cell for  density fields.  SCFT models have the capability to relax the size and shape of the simulation domain to find the optimal cell. We show next that our predictive ML has also this capability, by
preserving the relationship between the  Hamiltonian and the  density fields, and can be effectively employed for this computationally intensive task. 
 
  The problem can be defined as follows: for fixed $\chi N^* ,f^*$, the  Hamiltonian and the  density fields are totally controlled by the cell parameters $(l_1, l_2,\theta)$. Taking $l=l_1=l_2$ as an example, for any $(l,\theta)$, we can obtain the corresponding  density field $\rho^{SCFT}$ and the  Hamiltonian $H^{SCFT}$ by direct SCFT simulation. Thus, we can get a heatmap $map^{SCFT}$ of  $H^{SCFT}$ as a function of the domain cell size with color representing the
  numerical value of $H^{SCFT}$.  Similarly, we can use our ML predictive tool, the 
  Sobolev space-trained $\text{CNN}_{Sbv}$ and the ScftGAN, to produce a predicted 
  heatmap $map^{ML}$ of the ML predicted  Hamiltonian $H^{ML}$.
   If  $map^{ML}$ coincides with $map^{SCFT}$, our ML framework has the capability to find the optimal cell size. It is important to note that in this downstream task of finding $map^{ML}$, we do not use any SCFT tools/data after training. 

To increase the robustness of the predicted heatmap, we obtain the final $map^{ML}$ by taking the average of 20 predicted maps on the same domain.  To further assess the 
effectiveness of the end-to-end ML method we also compute a heatmap $map^{MIX}$that uses a mix of SCFT and ML by evaluating $\tilde{H}(\rho^{SCFT})$, where $\rho^{SCFT}$ is the density field computed by direct SCFT,  and $\tilde{H}$ is the ML   Hamiltonian trained in Sobolev space. Figure~\ref{fig:heatmap} compares 
 $map^{SCFT}$ (first column),  $map^{MIX}$ (second column), and  $map^{ML}$ (third column). The remarkable agreement of $map^{ML}$ with $map^{SCFT}$
 demonstrates the significant potential for the proposed ML methodology, based on the Sobolev-trained CNN and the ScftGAN, to be an effective end-to-end substitute of the computationally intensive SCFT computations required for this and other parameter-space exploration tasks. 

\begin{figure}[!htbp]
    \centering
    \includegraphics[width=\linewidth]{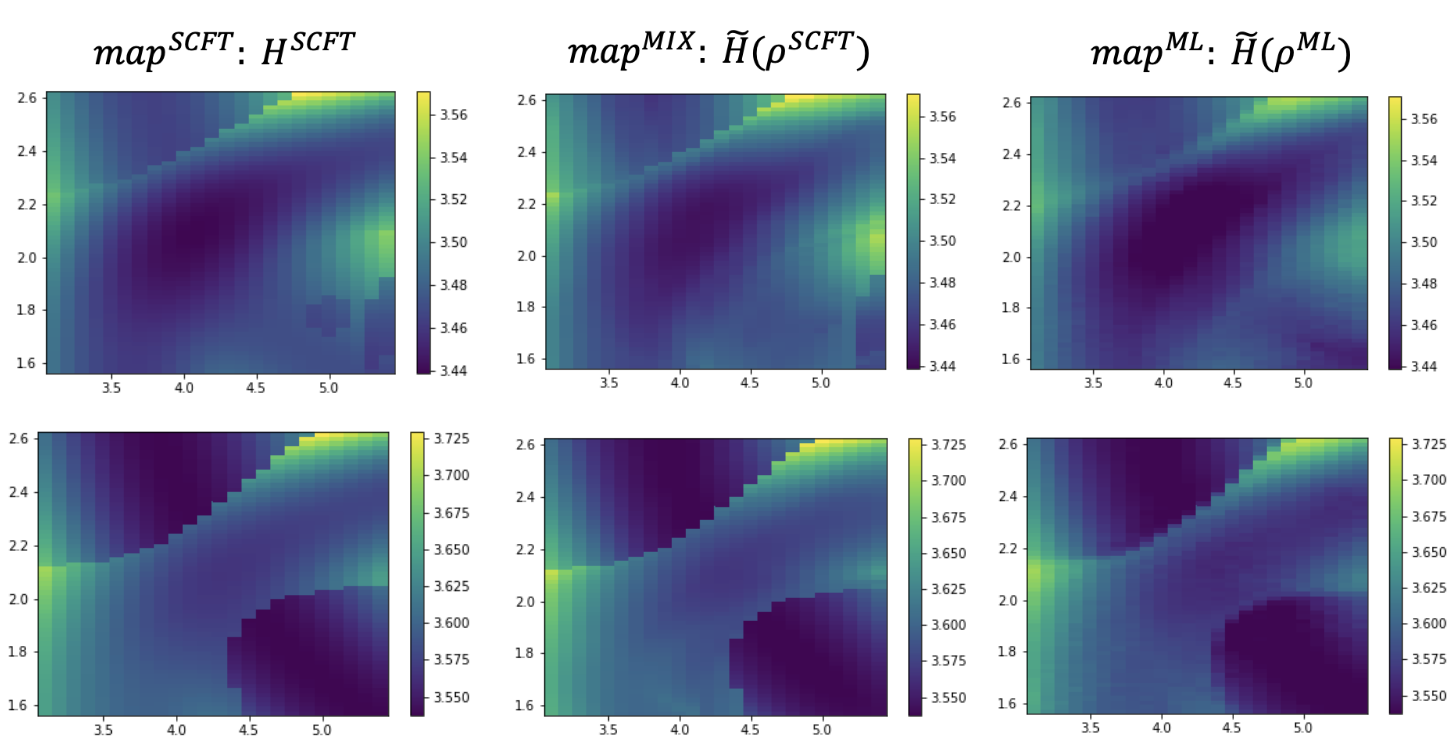}
    \caption{Heat maps \revision{(flooded contour plots) of the field-theoretic Hamiltonian with respect to cell size $(l, \theta)$ for two volume fractions $f$ of component $A$ of the diblock, $f=0.35$ (first row) and $f=0.42$ (second row).
     In each heat map, the horizontal axis is $l$ and the vertical axis is $\theta$. The color map represents the value of the Hamiltonian. The first column corresponds to the direct SCFT value. The second column corresponds to $\tilde{H}(\rho^{SCFT})$, where $\tilde{H}$ is the surrogate map learned with the convolutional network and $\rho^{SCFT}$ is the average monomer density (of blocks A) obtained from SCFT. Finally, the heat maps in the third column correspond to the end-to-end machine learning-based $\tilde{H}(\rho^{ML})$, where the predicted density field $\rho^{ML}$ was produced with the combined  ScftGAN and the convolutional network.}}
    \label{fig:heatmap}
\end{figure}  

\subsection{Efficient algorithm for cell size relaxation}\label{efficient_heatmap}
Section \ref{heatmap_screening} is a proof of concept to demonstrate the pre-trained ScftGAN is a powerful tool for cell size analysis and relaxation. Finding optimal cell shape/size is an important task in polymer design. We could simply screen the cell parameter space to find the optimal cell size minimizing the Hamiltonian as shown in Figure \ref{fig:heatmap}. However, the parameter space grows exponentially as the number of parameters increases. In such high dimensional space, direct searching for an optimal cell size is computational expensive. However, since the surrogate ScftGAN and $\tilde{H}$ are pre-trained, we could actually equip them with efficient searchers to optimize cell size. In this section, we consider a general three-dimensional space of $l_1, l_2,\theta$ ($l_1$ and $l_2$ are not necessarily equal) and propose to find the optimal cell size based on gradient descent method. Our goal is to efficiently solve  the following optimization problem 
\begin{equation}\label{cell_optmize}
    \argmin_{l_1,l_2,\theta} E_{z\sim p(z)}[\tilde{H}(l_1,l_2,\theta,\chi N,f,G(l_1,l_2,\theta,\chi N,f,z))].
\end{equation}

In the rest of the section, we denote the cell size parameters by $\mathbf{c}=(l_1,l_2,\theta)$. In each step of the gradient descent method to find a solution of (\ref{cell_optmize}), the expectation is estimated by Monte Carlo sampling of the noise $z$, and thus the gradient iterations are given by
\begin{equation}\label{cellsize:sgd}
    \mathbf{c}^{n+1}=\mathbf{c}^n-\frac{1}{\revision{N_e}}\sum_{i=1}^{\revision{N_e}} \nabla\tilde{H}(\mathbf{c}^n,\chi N,f,G(\mathbf{c}^n,\chi N,f,z^n_i)),
\end{equation}
where ${\revision{N_e}}$ is the number of random samples drawn for expectation estimation. In a bounded domain of interest, we use projected gradient descent method to ensure the searching happens in the interested domain.

We remark that in each step we update $\mathbf{c}$ based on sampling a mini-batch of size $N$ of $z^n$  rather than on a fixed pool of $z$, so the algorithm is equivalent to mini-batch stochastic gradient descent method.

To enhance the chance of finding the global optimum in the domain, we sample $M$ initial parameter candidates $\mathbf{c}^0$ in the domain, i.e., $\{\mathbf{c}_1^0,...\mathbf{c}_M^0\}$ and then update them based on iterative scheme defined in (\ref{cellsize:sgd}) to get $M$ candidates of optimal cell size $\{\tilde{\mathbf{c}}_1,...,\tilde{\mathbf{c}}_M\}$. Finally, the prediction is obtained by selecting the candidate which generates the smallest Hamiltonian value, i.e.,

$$\argmin_{\mathbf{c}\in \{\tilde{\mathbf{c}}_1,...,\tilde{\mathbf{c}}_M\} } E_{z\sim p(z)}[\tilde{H}(\mathbf{c},\chi N,f,G(\mathbf{c},\chi N,f,z))],$$
where the expectation is again estimated by Monte Carlo Sampling.

We just introduced a method for global searching in the cell size domain. We could also search in any sub-domain of cell size with a specific symmetry. For instance, an important families of cells are defined by the cells with the "p6mm" symmetry where $l_1=l_2$ and $\theta=\frac{2\pi}{3}$. If we want to search for the optimal cell size within the symmetry family, we could simply search for 

$$\argmin_{l} E_{z\sim p(z)}[\tilde{H}(l,l,\frac{2\pi}{3},\chi N,f,G(l,l,\frac{2\pi}{3},\chi N,f,z))],$$
where the the gradient could be computed by the chain rule in gradient descent iterations similar to (\ref{cellsize:sgd}).

We report now on a comparison between direct SCFT and the proposed ML-method predicting the optimal cell size. We tested both the global and the p6mm optimal cell size predictions. Since there is randomness in the searching process, we report the average of the absolute error in 10 independent searches. The domain of interest is defined by $l_1\in [3,5.5], l_2\in [3,5.5], \theta \in[\pi/2, 5\pi/6]$.

\begin{table}[htp!]
\centering
\begin{tabular}{cccc}
\hline
Avg difference & $l_1$  & $l_2$  & $\theta$ \\ \hline
Global    & 0.0149 & 0.0866 & 0.0261   \\ \hline
p6mm      & 0.0390 & 0.0390 & -        \\ \hline
\end{tabular}
\caption{Average difference between SCFT-predicted optimal cell size and ML-predicted optimal cell size. The average difference is computed over 10 independent ML predictions. ``Global" is the case we search among all the possible cells in the domain and ``p6mm" is the case we search among cells satisfying $l_1=l_2$ and $\theta=2\pi/3$ (so there is no need to compare difference of $\theta$ in p6mm).  $\chi N=16, f=0.35, M=4, N=4$. }
\label{optimal_cell_search}
\end{table}

As Table \ref{optimal_cell_search} shows the predictions of optimal cell size have 2 digits of accuracy. In the heatmap of Fig.~\ref{fig:heatmap}, there is generally a region of good cell size/shape candidates near the local/global optimal, which demonstrates that 2 digits of accuracy in cell optimization is sufficient for generating cell shape/size design candidates. 

We intend to explore other global optimum searching algorithms, such as Bayesian optimization or the Monte Carlo method,  to be built on top of our proposed ML-based methodology.

\subsection{Computational Efficiency}
To evaluate the computational efficiency of the proposed ML tool for predicting 
 density fields, we design the following numerical experiment. 
 We sample 5000 combinations of parameters in the domain of $l_1\in [5.1,5.5]$, $l_2 \in [4.6,5.5]$, $\theta \in [\pi/2,\frac{5\pi}{6}]$ and $f\in [0.41,0.5]$
 and run the  SCFT code and the ML tool on a CPU for these parameters. We also run 
 the ML tool on a GPU cluster to further assess its efficiency.
 The SCFT and ML CPU times are measured on the same machine (MacBook Pro, 2.2 GHz Intel Core i7 Processor, 16 GB 1600 MHz DDR3 Memory). The ML GPU time is reported for a Tesla P100 GPU. Figure~\ref{fig:timecomparison} shows the computational time to get the density fields on a 1k $\sim$ 5k combination of parameters. Clearly, the ML approach, on either CPU or GPU is significantly more efficient that the direct SCFT approach and its computational time grows much more slowly as the number of examples increases. This provides further evidence of the potential of proposed ML methodology to  accelerate the exploration of parameter space and phase discovery. 
\begin{figure}[!htp]
    \centering
    \includegraphics[width=.9\linewidth]{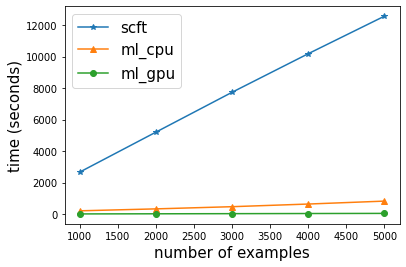}
    \caption{Computational time for \revision{the direct self consistent field theory (SCFT) computation of the average monomer density field and that by the machine learning model for a sample of 5000 combinations of parameters of cell size and shape, $l_1\in [5.1,5.5]$, $l_2 \in [4.6,5.5]$, $\theta \in [\pi/2,\frac{5\pi}{6}]$, and of volume fraction, $f\in [0.41,0.5]$. For the CPU comparison the same computer was employed. The GPU times were obtained on a Tesla P100 GPU cluster.}}
    \label{fig:timecomparison}
\end{figure}

\begin{section}{Conclusions}\label{sec:2d_con}
In this paper, we propose a ML approach to solve two central problems for 
polymer phase discovery via SCFT: accurate evaluation of the  Hamiltonian $H$ and density fields at saddle points. The proposed ML methodology consists of a CNN trained in Sobolev space to simultenously approximate $H$ and its gradient
with rotation and translation invariance, and an GAN-based algorithm to generate 
pools of density field candidates from the physical parameters 
for an ultimate selection of an optimal field with the Sobolev space-trained CNN.

Numerical experiments demonstrate both the efficacy and efficiency of the proposed ML methodology and show that its two main components can work in concert to solve the downstream task of finding the optimal cell size, either globally or within specific symmetry.
It is important to note that the Sobolev space-trained CNN  yields an accurate approximation of $H$ (and its gradient in a vicinity of saddle points) that is remarkably fast to evaluate.  Furthermore, we could also equip the proposed CNN with other screening methods,  such as Markov chain Monte Carlo (MCMC) or the Hamiltonian Monte Carlo algorithm, to enhance its prediction of saddle point density fields. 
 
Our methodology was developed for the 2D spatial setting. However, since 3D data can be viewed as an image with an additional channel, we believe the proposed CNN-based method can also be employed in this higher dimensional setting. 
 
\end{section}

\begin{section}{Acknowledgements}
H.D.C. and Y.X. acknowledge partial support from the National Science Foundation under award DMS-1818821. G. H. F. and K. T. D. were supported by the CMMT Program of the National Science Foundation (NSF) under Grant No. DMR-2104255. Use was made of
computational facilities purchased with funds from the NSF (OAC-1925717) and administered
by the Center for Scientific Computing (CSC). The CSC is supported by the California NanoSystems Institute and the
Materials Research Science and Engineering Center (MRSEC;
NSF DMR 1720256) at UC Santa Barbara.  
The authors appreciate Prof. Ruimeng Hu's support in computational resources.  
\end{section}

\clearpage

{\small
\bibliography{main}}
\bibliographystyle{unsrt}

\clearpage
\appendix
\section{Appendix}
\subsection{Early stopping in training ScftGAN}
We apply an early stopping technique to increase the stability in training the ScftGAN.
In order to ensure that the G outputs are not too far from the vicinity of the data manifold and the trained-ScftGAN is consistent with the Sobolev-trained CNN, we evaluate  the performance $\tilde{H}(G(l_1,l_2,\theta, f))$ on the validation set and stop at the epoch where the error satisfies the stopping criterion below. We start the evaluation in the last $T_e$ epochs and only do the evaluation in partial epochs to save computational time. The training will stop at,
\setcounter{equation}{0}
\begin{equation}
    \renewcommand{\theequation}{\thesection.\arabic{equation}}
    \argmin_{e \in e_{sample}} \sum_{i=1}^{N_V}\frac{1}{N_V}\|\tilde{H}(G(l_{1i},l_{2i},\theta_i, f_i))-H_i\|^2<\epsilon,
\end{equation}
where $N_V$ is the size of validation set, $e_{sample} \subseteq [T-T_e, T]$ with $T$ as the maximal number of epochs and $e_{sample}$ are sampled from the last $T_e$ epochs. In experiments, we sample $e_{sample}$ and evaluate the performance on the validation set about every 7 epochs (after training on every 2000 batches).

\subsection{Hyperparameter tuning}\label{appendix:hyperparameter}

The architecture of the Sobolev space-trained CNN is shown in Eq.~(\ref{eq:architecture}). We tune the hyperparameters by greedy search but  not heavily. We list in Table~\ref{hyper_cnn}
the hyperparameters used in the Sobolev space-trained CNN. $T$ is the number of epochs, $lr$ is the learning rate, $\beta$ is coefficient in the loss function (\ref{2d:costnn}), and $batchsize$ is the number of instances in each batch during training.  

\begin{table}[htp]
\centering
\begin{tabular}{cccc}
\hline
$T$    & $lr$   & $\beta$  & $batchsize$ \\ \hline
5000 & $5\times 10^{-4}$ & $1/32^2$ & 256       \\ \hline
\end{tabular}
\caption{Hyperparameter in the Sobolev space-trained CNN.}
\label{hyper_cnn}
\end{table}
In the architecture of Sobolev space-trained CNN, the number of channels in the output of each convolutional layers in turn are $(8, 16, 32, 8, 8)$, the number of neurons in each fully connected layers are $(64, 32, 16, 1)$.

The architecture of the ScftGAN is defined by Eq.~(\ref{scftgan_generator}) and Eq.~(\ref{scft_discriminator}). The hyperparameters in the ScftGAN are listed in Table \ref{hyper_Scftgan}. As in the Sobolev space-trained CNN, $T$ is the number of epochs, $lr$ is learning rate, $batchsize$ is the number of instances in each batch, $T_e$ is the number of epochs applying early-stopping criterion, $\epsilon$ is the threshold value in early-stopping criterion, $\lambda$ is the coefficient in loss function (\ref{scftgan_loss}) and $nz$ is the dimension of random noise in the ScftGAN.

\begin{table}[htp]
\centering
\begin{tabular}{ccccccc}
\hline
$T$  & $lr$ & $T_e$ & $\epsilon$ & $batchsize$ & $\lambda$ & $nz$ \\ \hline
6000 & $5\times 10^{-5}$ & 100  &$10^{-4}$ & 256         & $3 \times 10^4$       & 16   \\ \hline
\end{tabular}
\caption{Hyperparameters in the ScftGAN.}
\label{hyper_Scftgan}
\end{table}

In the generator $G$, the number of channels in the output of transposed convolutional layers are $(512, 256, 128, 64, 64, 64, 64, 64, 64, 64, 1)$. In the discriminator $D$, the number of channels in the output of each convolutional layers are $(64,128,256,512,1)$.

\subsection{\revision{An analysis of the dependence of the CNN Hamiltonian prediction error on the training set size}}\label{train_set_size_analysis}

\revision{In this section, we investigate the impact of the training set size on the prediction error of the Sobolev space-trained convolutional neural network in predicting the Hamiltonian.  To achieve this, we randomly select subsets from the training set with sizes that are powers of 2, such as 64, 128, 256, 512, etc. Given the stochastic nature of the gradient descent process, we train the CNN three times independently for each sub-training set of a particular size. We then compute the average root square error over the three runs on both the test set and training set, and the results are plotted in Fig.~\ref{fig:hvstrainsize}.}

\revision{Our findings indicate that increasing the number of data points in the training set leads to a decrease in both the training error and test error. Moreover, the gap between training error and test error gradually narrows down. There is a relatively larger gap between training error and test error when the training set size is as small as 64 ($2^6$). However, this gap is fully closed when the training size grows to 4096 ($2^{12}$). Nonetheless, when $N_T \geq 256$, both the training error and the test error reach an acceptable level.}

\revision{In this numerical experiment, we randomly selected the sub-training set, but there might be more effective sampling strategies that could further enhance model performance on small data. We believe that a  
study of small-data sampling is a meaningful problem that requires further exploration.}

\begin{figure}[!htp]
    \centering
    \includegraphics[width=.9\linewidth]{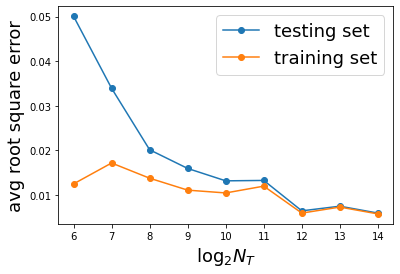}
    \caption{\revision{Hamiltonian approximation error versus training set size on the training set and the test set. The $x$-axis is log based 2 of $N_T$. The $y$-axis is the average root square error over 3 repeated runs.}}
    \label{fig:hvstrainsize}
\end{figure}  
\end{document}